\documentclass[runningheads]{llncs}
\usepackage[T1]{fontenc}
\usepackage{graphicx}
\usepackage{algorithm}
\usepackage{algpseudocode}

\usepackage[T1]{fontenc}
\usepackage{graphicx}

\usepackage{amsmath,amssymb,amsfonts}
\usepackage{stmaryrd}
\usepackage[notextcomp]{stix}
\usepackage{times}
\usepackage{enumerate} 
\usepackage{todonotes}
\usepackage{tabularx}
\usepackage{multicol}
\usepackage{array}
\newcommand{\todoPL}[2]{\todo[color=yellow!50,#1]{\textbf{PL:}#2}}
\newcommand{\todoSZ}[2]{\todo[color=cyan,#1]{\textbf{SZ:}#2}}

\begin{document}
\title{Learning Shuffle Ideals with Membership Queries and Contrastive Examples}
%
%
\author{S.\ Mahmoud Mousawi\inst{1}\orcidID{0009-0002-8689-0560} \and
Pierluigi San Pietro\inst{2}\orcidID{0000-0002-2437-8716} \and
Sandra Zilles\inst{1,3}\orcidID{0000-0001-7834-8574}}
\authorrunning{S.M.\ Mousawi et al.}
%
\institute{University of Regina, Canada 
\email{mousawi.s.m@gmail.com,sandra.zilles@uregina.ca}\\
\and
DEIB-Politecnico di Milano and CNR-IEIIT, Italy\\
\email{pierluigi.sanpietro@polimi.it}
\and
Alberta Machine Intelligence Institute, Canada}
\maketitle              
\begin{abstract}
This paper studies learning of shuffle ideals with membership queries as well as with contrastive queries---a form of membership query that reveals not only whether a selected word $w$ is in the target language or not, but also provides a most similar word $w'$ that belongs to the target language if 
$w$ does not. 

        For both settings, it is shown that even some very simple classes of shuffle ideals cannot be learned efficiently. By contrast, we obtain positive learnability results for classes of shuffle ideals that meet certain structural conditions. In the case of membership queries, these structural conditions are related to the previously studied notion of universal words, and raise new questions in word combinatorics. In the case of contrastive queries, the structural conditions relate to partitioning sets of words that are listed in shortlex order.
\end{abstract}
%
%
%

%
\section{Introduction}
 
Due to their applications in modeling and processing structured data, classes of regular languages have been studied in computational learning theory for decades, in the context of PAC-learning \cite{DBLP:journals/cacm/Valiant84}, learning in the limit \cite{DBLP:journals/iandc/Gold67}, and learning from queries \cite{DBLP:journals/ml/Angluin88}. 

While the class of all regular languages is not (passively) learnable in the limit from only positive examples  \cite{DBLP:journals/iandc/Gold67} or (actively) learnable from membership queries, more encouraging results have been obtained with other forms of queries. For instance, it can be efficiently learned from a combination of membership and equivalence queries \cite{DBLP:journals/iandc/Angluin87}, and even from equivalence queries alone if the counterexamples provided to the learner follow a certain preference order \cite{DBLP:journals/jcss/IbarraJ91}. Here a membership query consists of a word $w$ of the learner's choice; the answer will tell the learner whether or not $w$ belongs to the language $L$ to be learned. Similarly, an equivalence query is a language $L'$; the learner will be told whether $L'=L$. In the negative case, the learner also receives a counterexample, i.e., a witness from the symmetric difference of $L$ and $L'$.

Since equivalence queries are not always 
practical, many learning-theoretic studies focus on (1) potentially easier to learn subclasses of regular languages, such as reversible languages \cite{DBLP:journals/jacm/Angluin82}, piecewise testable languages \cite{DBLP:journals/grammars/GarciaR04}, locally testable languages \cite{DBLP:journals/grammars/GarciaR04}, and regular pattern languages \cite{DBLP:journals/tcs/NesselL05}, as well as (2) variations on membership queries that promise more learning ability, such as various forms of correction queries \cite{DBLP:conf/icgi/Becerra-BonacheDT06}.

In terms of aspect (1), the focus of our work is on a seemingly very simple subclass of regular languages, namely the so-called \emph{shuffle ideals}~\cite{Sakarovitch_Imre_1997}. A shuffle ideal is defined by a non-empty finite set of words called generators, and consists of all words $w$ that contain at least one of the generators $u$ ($=u(1)u(2)\ldots u(n)$) as a scattered subword, in the sense that $w$ can be written as $\Sigma^* u(1) \Sigma^* u(2) \Sigma^* \dots \Sigma^* u(n) \Sigma^*$.\footnote{Throughout this paper, we assume $\Sigma$ is a fixed finite alphabet of size at least 2.} The Boolean closure of the class of shuffle ideals is known as the class of piecewise testable languages~\cite{DBLP:conf/automata/Simon75}.

As for aspect (2), we study learning of shuffle ideals both in the plain membership query setting and in the setting of \emph{contrastive queries}\/~\cite{mansouri2026formal}. A contrastive query is a membership query $w$, in which the answer is enhanced by a contrastive example $w'$. Here $w'$ has opposite label from $w$, meaning that one of the two words belongs to the target language $L$, while the other does not. Moreover, with respect to a fixed distance function $d$ over $\Sigma^*$, the word $w'$ must be among the closest words to $w$ that have opposite label from $w$. Such a contrastive example conveys to the learner the additional information that all words between $w$ and $w'$ (in terms of $d$) have the same label as $w$. This information is very helpful in learning in many settings~\cite{mansouri2026formal}, and makes contrastive queries stronger than correction queries (in which contrastive examples are only given for queries $w\notin L$~\cite{DBLP:conf/icgi/Tirnauca08}). Our paper is the first to study learning of formal languages under this notion of contrastive queries.

Previous research on the learnability of classes related to shuffle ideals focused on learning from random examples or identification in the limit from positive data, rather than query learning~\cite{DBLP:journals/jmlr/AngluinAEK13,DBLP:journals/grammars/GarciaR04}. These works considered very limited settings of shuffle ideals with a single generator or piecewise testable languages built from  generators  whose length is bounded by a constant.
The reason for these limitations is that shuffle ideals, despite their very simple definition, are very difficult to learn. We show, among other results, that shuffle ideals can be learned efficiently with membership queries, if they are known to have exactly two generators, but not if they are known to have exactly three generators. In doing so, we reveal a connection to an interesting word-combinatoric problem based on nearly universal words~\cite{DBLP:conf/dcfs/FleischmannHHMN22}. A nearly $n$-universal word is a word $s$ that contains all words of length $n$ as scattered subwords, except one specified such word; Fleischmann et al.\ also extend this notion to a larger set of forbidden words of length $n$. We generalize their notion of near-universality to a word-combinatoric notion we call $(\Sigma^n\setminus W)$-universal sets. Roughly speaking, the computability of ``small'' $(\Sigma^n\setminus W)$-universal sets is a sufficient condition for efficient learning of subclasses of shuffle ideals with membership queries. When exactly our sufficient condition for learnability is met is an open problem that thus translates into a nice word-combinatoric question.

A second part of our paper studies the potential gain one can achieve from contrastive queries (compared to classical membership queries), when learning classes of shuffle ideals. We show 
that the class of all shuffle ideals for which 
all generators have the same length can be learned from a polynomial number of contrastive queries. Our strategy is based on a shortlex distance function and makes use of the fact that a first generator is always easy to find. This tells the learner the length of the remaining generators. With the shortlex ordering, previously found generators then partition the set of potential further generators into intervals w.r.t.\ the shortlex order. These intervals can then be searched efficiently using contrastive queries. This interval search strategy 
extends to other non-trivial classes of shuffle ideals.

In sum, our work is the first to provide a detailed study on learning shuffle ideals with membership queries, and opens up new avenues for positive yet non-trivial learning results by invoking contrastive queries. The new purely word-combinatoric problem defined in this context is interesting in its own right, 
beyond the field of learning theory.


\section{Preliminaries}

Throughout this paper, we fix a finite alphabet $\Sigma$ of size at least 2, and denote the set of all finite words over $\Sigma$ by $\Sigma^*$. For any $n\in\mathbb{N}$, we use $\Sigma^{\ge n}$, $\Sigma^{\le n}$, and $\Sigma^n$, resp., to denote the set of all words (over $\Sigma$) of length at least $n$, at most $n$, and equal to $n$, resp. We also use $\Sigma^*=\Sigma^{\ge 0}$ and $\Sigma^+=\Sigma^{\ge 1}$. If $w\in\Sigma^*$ and $a\in\Sigma$, then $|w|_a$ denotes the number of occurrences of the letter $a$ in $w$.
If $\mathcal{L}$ is any class of languages over $\Sigma^*$, then $\overline{\mathcal{L}}$  denotes the class of all complements of languages in $\mathcal{L}$. 

\subsection{Shuffle Ideals}

Let $\sqsubseteq$ be the binary relation on $\Sigma^*$ such that $u \sqsubseteq v$ if $u$ is a (possibly non-contiguous) subword of $v$; we say $u$ is a \emph{scattered subword}\/ of $v$. Denote by $\shuffle$ the shuffle operation, defined on both words and languages.
A language $L \subseteq \Sigma^*$ is a \emph{shuffle ideal} if $L \shuffle \Sigma^* = L$. Every shuffle ideal is finitely generated: $L$ is a shuffle ideal if and only if $L = U \shuffle \Sigma^*$ for some finite set $U$. 
The \emph{principal shuffle ideal} generated by $u \in \Sigma^+$ is the language
\[
SI(u) = \{x \in \Sigma^* \mid u \sqsubseteq x\} = \Sigma^* u(1) \Sigma^* u(2) \Sigma^* \dots \Sigma^* u(n) \Sigma^*,
\]
where $u = u(1) u(2) \dots u(n)$ and $n = |u|$. The value $n$ is called its order.

By Higman's Lemma \cite{Higman52}, every shuffle ideal is a finite union of principal shuffle ideals. The \emph{shuffle ideal} generated by a finite set $U \subseteq \Sigma^+$ is 
$SI(U) = \bigcup\nolimits_{u \in U} SI(u)$.  

The words in $U$ are called \emph{generators}. The order of $U$ and of $SI(U)$ is the length of the longest word of $U$. 
We assume that $U$ is \emph{reduced}: there are no distinct words $u,v$ in $U$ such $u\sqsubseteq v$ ($v$ is redundant, since it has $u$ as a scattered subword).
The family of all shuffle ideals, denoted by $\mathcal{SI}$, is closed under union and intersection but not under complement: it forms a positive variety, generated by principal shuffle ideals.

The complexity of algorithms for learning shuffle ideals (or their complements) is assessed in terms of the size of the underlying target concept. For a reduced generator set $U$, we measure the size of the target concept $SI(U)$ by $
\|U\|=\sum_{u\in U} |u|.$ 
All complexity bounds below are understood with respect to this parameter.


\subsection{Query Learning -- Definitions and Preliminary Results}

Given a class $\mathcal{L}$ of languages over an alphabet $\Sigma$, the goal of a learner is to identify an unknown language $L\in\mathcal{L}$ 
by asking queries about $L$. A query learner is a deterministic procedure interacting with an oracle in sequential rounds, each time asking a query of a specified type, until the learner makes a conjecture $L'$, which ends the learning process. 

The first type of query we consider is a \emph{membership query}~\cite{DBLP:journals/ml/Angluin88}, which the learner formulates as a word $w\in\Sigma^*$. The oracle's answer is $1$ if $w\in L$ for the target language $L$, and $0$ otherwise; we denote this Boolean indicator by $L(w)$.
The language $L$ is successfully identified if $L=L'$ for the learner's conjecture $L'$. 

We will study the limitations of membership queries when learning classes of shuffle ideals. In this context, we  contrast membership queries with a stronger form of query---so-called \emph{contrastive queries}. 
In the contrastive query setting,  the learner and the oracle agree on a distance function $d:\Sigma^*\times\Sigma^*\rightarrow\mathbb{R}^{\ge 0}$ (not necessarily a metric), such that, for any $w\in\Sigma^*$, the set $\{d(w,w')\mid w'\in\Sigma^*\}$ is discrete. In any one round of interaction, the learner asks a query represented by a word $w\in\Sigma^*$. 
The oracle's answer belongs to $\{0,1\}\times(\Sigma^*\cup\{\bot\})$. If $w\in L$, the oracle responds with $(1,w')$, where $w'\notin L$. If $w\notin L$, the oracle responds with $(0,w')$, where $w'\in L$. In either case, the word $w'$ (called \emph{contrastive example}\/) is chosen such that $w'\in\arg\min_{L(w'')\ne L(w)}d(w,w'')$. (In case there is no $w''\in\Sigma^*$ with $L(w'')\ne L(w)$, the oracle sets $w'=\perp$.)\footnote{Contrastive queries are a generalization of \emph{correction queries}\/ \cite{DBLP:conf/icgi/Becerra-BonacheDT06,DBLP:conf/icgi/Tirnauca08}. The response to a correction query $w$ includes a word $w'$ iff $w\notin L$ (and thus $w'\in L$). That is, a word of opposite label is provided only if the learner queries a word outside the target language. Contrastive queries, however, always provide the learner with an example of opposite label (if one exists).} 

Note that the answer $(\ell,w')$ to a contrastive query $w$ provides the learner with more information than just information on the membership of $w,w'$ in the target language $L$. Even if a language $L'$ satisfies $L'(w)=L(w)\ne L'(w')=L(w')$, the learner can still exclude $L'$ from its pool of candidates for $L$ in case $w'\notin\arg\min_{L'(w'')\ne L'(w)}d(w,w'')$.


A learner learns the class $\mathcal{L}$ with (polynomially many) queries if it successfully identifies every $L$ in $\mathcal{L}$ (with a number of queries polynomial in the size of the underlying representation of $L$ and the length of the longest response provided by the oracle during learning). 
Every class of languages that is learnable from (polynomially many) membership queries is also learnable from (polynomially many) contrastive queries, but the opposite is not true, as we will see in Example~\ref{example:memq-conq}.

It is well-known (and trivial) that a class $\mathcal{L}$ is learnable from membership queries iff $\overline{\mathcal{L}}$ is. The same is obviously true for contrastive queries (see the appendix).

\begin{proposition}\label{prop:complement}
    Let $\mathcal{L}$ be any class of languages over $\Sigma^*$ and $d$ a distance function. Then $\mathcal{L}$ is learnable from (polynomially many) contrastive queries w.r.t.\ $d$ iff $\overline{\mathcal{L}}$ is learnable from (polynomially many) contrastive queries w.r.t.\ $d$ (where the size of $\overline{SI(U)}$ is $\|U\|$). 
\end{proposition}

The following is an immediate consequence of Proposition~\ref{prop:complement}.

\begin{corollary}\label{cor:complement}
    Let $\mathcal{L}$ be any class of languages over $\Sigma^*$ and $d$ be a distance function. Suppose $\mathcal{L}$ is learnable from (polynomially many) contrastive queries w.r.t.\ $d$. Suppose further that there is an algorithm that, given access to an oracle for an unknown language $L\in\mathcal{L}\cup\overline{\mathcal{L}}$, asks (polynomially many) contrastive queries w.r.t.\ $d$ and then correctly determines whether $L\in\mathcal{L}$. Then $\mathcal{L}\cup\overline{\mathcal{L}}$ is learnable from (polynomially many) contrastive queries w.r.t.\ $d$.
\end{corollary}

\begin{definition}
    Let $d:\Sigma^*\times\Sigma^*\rightarrow\mathbb{R}$ be any distance function over $\Sigma^*$ (not necessarily a metric).  We say that $d$ is length-compatible
    if, for any words $w_1$, $w_2$, $w_3$ in $\Sigma^*$, we have
    \[
    \bigl||w_3|-|w_1|\bigr|>\bigl||w_2|-|w_1|\bigr|\ge 0\Rightarrow d(w_1,w_3)>d(w_1,w_2)\,.
    \]
\end{definition}

\begin{example}\label{exmp:length-compatible}
\begin{enumerate}
    \item The function $d_0$ defined by $d_0(w,w')=\bigl||w|-|w'|\bigr|$ is length-compatible. 
    \item The edit distance, denoted by $d_e$, is not length-compatible; for instance $d_e(aa,a)=1<2=d_e(aa,bb)$, although $|aa|-|a|=1>0=|aa|-|bb|$.
\end{enumerate}
\end{example}

We are now ready to provide an example showing that contrastive queries are strictly more powerful than membership queries.

\begin{example}\label{example:memq-conq}
    Let $\mathcal{L}=\{\Sigma^{\ge n}\mid n\in\mathbb{N}\}\cup\{\emptyset\}$. Then $\mathcal{L}$ is not learnable from membership queries, but (i) learnable from polynomially many contrastive queries w.r.t.\ any distance function and (ii)~learnable from a single contrastive query w.r.t.\ any length-compatible distance function. 

    \end{example}
    \begin{proof}
    By way of contradiction, suppose a learner $A$ learns $\mathcal{L}$ from membership queries. For the target language $\emptyset$, every query of $A$ will be answered with 0, until eventually $A$ must halt and conjecture $\emptyset$. Let $n_0$ be the length of the longest word $A$ has queried up to that point. Then $A$ cannot learn any language $\Sigma^{\ge n}$ for $n> n_0$.

    A learner using contrastive queries will first query the word $\varepsilon$. If the answer is positive, the learner conjectures $\Sigma^*=\Sigma^{\ge 0}$. If the answer is $(0,\perp)$, the learner conjectures $\emptyset$. Otherwise, the answer is of the form $(0,w)$ where $w$ belongs to the target language, which equals $\Sigma^{\ge n^*}$ for some $n^*\le |w|$. The learner will now use logarithmically (in $|w|$) many queries of the form $a^n$ in order to identify $n^*$, by performing binary search on the range $\{0,\ldots,|w|\}$. This way, the learner succeeds regardless of the distance function.

    Under any length-compatible distance function, the word $w$ in the above reasoning must be of length $n^*$, so that the learner can immediately conjecture $\Sigma^{\ge |w|}$. \qed
\end{proof}

\begin{definition}
    Let $d_1,d_2:\Sigma^*\times\Sigma^*\rightarrow\mathbb{R}$ be two functions. We say that $d_1$ is a refinement of $d_2$ iff, 
        $\forall w_1,w_2,w_3\in\Sigma^*\quad d_2(w_1,w_3)>d_2(w_1,w_2)\Rightarrow d_1(w_1,w_3)>d_1(w_1,w_2)$.
\end{definition}

Proposition~\ref{prop:refinement} (proven in the appendix) and Corollary~\ref{cor:refinement} are similar in spirit to Remark 4/Corollary 9 in \cite{mansouri2026formal}.

\begin{proposition}\label{prop:refinement} Let $\mathcal{L}$ be a class of languages and $d_1,d_2:\Sigma^*\times\Sigma^*\rightarrow\mathbb{R}$ two functions such that $d_1$ is a refinement of $d_2$. If $\mathcal{L}$ can be learned from (polynomially many) contrastive queries  using the distance function $d_2$, then it can be learned from (polynomially many) contrastive queries using the distance function $d_1$.
\end{proposition}

Since any length-compatible distance function refines the function $d_0$ from Example~\ref{exmp:length-compatible}, we obtain:

\begin{corollary}\label{cor:refinement} Let $\mathcal{L}$ be a class of languages. If $\mathcal{L}$ can be learned from (polynomially many) contrastive queries using the distance function $d_0$, then it can be learned from (polynomially many) contrastive queries using any  length-compatible distance function.
\end{corollary}

\section{Learning Shuffle Ideals from Membership Queries}

Our study ties the problem of learning shuffle ideals with membership queries to the problem of finding what we call \emph{$(\Sigma^n\setminus W)$-universal sets} for any $n\in\mathbb{N}$ and some finite sets $W\subseteq \Sigma^*$. 
Such $S$ has the property that every length-$n$ word is a scattered subword of some word in $S$---while no word from $W$ is a scattered subword of any word in $S$.

\begin{definition}
Let $n\in\mathbb{N}$ and $W,S\subseteq\Sigma^*$. We say $S$ covers $v\in\Sigma^*$ if $v\sqsubseteq s$ for some $s\in S$. Further, $S$ is a  $(\Sigma^n\setminus W)$-universal set if the words in $\Sigma^n$ covered by $S$ are exactly those that do not cover any word in $W$, i.e., if $\Sigma^n\setminus\bigcup_{w\in W}\{v\mid w\sqsubseteq v\}\subseteq\bigcup_{s\in S}\{v\mid v\sqsubseteq s\}$ and $\bigcup_{s\in S}\{v\mid v\sqsubseteq s\}\cap \bigcup_{w\in W}\{v\mid w\sqsubseteq v\} =\emptyset$.
\end{definition}

To the best of our knowledge, $(\Sigma^n\setminus W)$-universal sets have previously been studied only for the case that (i) these sets are required to be singletons and, simultaneously, (ii)~$W\subseteq\Sigma^n$. In other words, one has tried to construct a \emph{single word}\/ that contains all length-$n$ words as scattered subwords, except for a set $W$ of \emph{length-$n$ words}\/~\cite{DBLP:conf/dcfs/FleischmannHHMN22}. It was shown that such a word always exists (and is of length linear in $n$) when $|W|\in\{0,1\}$, but only exists for binary alphabets when $|W|=2$~\cite{DBLP:conf/dcfs/FleischmannHHMN22}. By contrast, our notion calls for a \emph{set $S$}\/ of words that jointly contain all length-$n$ words as scattered subwords, except those length-$n$ words that themselves \emph{contain a word from $W$ as a scattered subword}.

To show the usefulness of $(\Sigma^n\setminus W)$-universal sets for learning shuffle ideals, we will deploy a basic procedure to be reused in several contexts. For any shuffle ideal $L$, when given a word $w\in L$ and a membership oracle for $L$, this procedure identifies a generator of $L$ that is a scattered subword of $w$. We dub this procedure \emph{distilling $w$ into a generator}; see Algorithm~\ref{alg:distilling}. Now the following lemma is straightforward.

\begin{algorithm}[t]
\caption{Distilling a positive witness $w\in L$ into a generator of $L$.} \label{alg:distilling}
\begin{algorithmic}[1] 
    \Procedure{Distill}{$w$}
        \State Ask queries for each subword of $w$ resulting from the removal of a single letter until receiving the answer 1 or until all such queries have been asked.
        \State If a query, say for $w'\sqsubseteq w$, is answered 1, then set $w:=w'$ and repeat from step~2.
    \State If all these queries are answered 0, then output $w$.
    \EndProcedure
\end{algorithmic}
\end{algorithm}

\begin{lemma}\label{lem:distilling}
If $L$ is a shuffle ideal, $w\in L$, and Algorithm~\ref{alg:distilling} accesses a membership oracle for $L$, then Algorithm~\ref{alg:distilling} asks $O(|w|^2)$ membership queries and outputs a generator of $L$ that is a scattered subword of $w$.
\end{lemma}

Finally, we can prove that efficient construction of ``small'' $(\Sigma^n\setminus W)$-universal sets would imply efficient learning of shuffle ideals from membership queries. In the following lemma, $M$ denotes a class of potential sets $U$ of generators that one may wish to learn. Intuitively, $W$ is a strict subset of an unknown target set $U$, representing the set of generators that a learner has already found at some step during the learning process. The index $n$ denotes the length of candidates for a new generator $u\in U\setminus W$ that the learner aims to find in the next step. In essence, the lemma states that a learner can find $u$ efficiently with membership queries, if it has access to a procedure that lists, for every length $n$, a $(\Sigma^n\setminus W)$-universal set $S(n,W)$---under the condition that the sum of the lengths of words in this set $S(n,W)$ is polynomial. The proof is in the appendix.

\begin{lemma}\label{lem:k-nearlyUniversalMemQ}
Let $M$ be a set of finite subsets of $\Sigma^*$, where no $U\in M$ contains words $u\!\ne\! u'$ with $u\!\sqsubseteq \! u'$. 
Suppose there is an algorithm with the following input/output behavior.\\
Input: $(n,W)$ where 
$n\le\max_{u\in U} |u|$ and  $W\subset U$ for some unknown $U\in M$.\\
Output: a $(\Sigma^n\setminus W)$-universal set $S(n,W)\subseteq \Sigma^{\ge n}$ such that $\sum_{s\in S(n,W)}|s|$ is bounded from above by a polynomial in $(n,|W|)$. 

Then there is a learner that, for any target language $SI(U)$ with $U\in M$, and on input of any $W\subset U$, finds a generator $u\in U\setminus W$, using polynomially many membership queries in $\|U\|$.
\end{lemma}

If a learner is not just trying to learn a single new generator, but learning the entire set $U$ of generators from scratch, then Lemma~\ref{lem:k-nearlyUniversalMemQ} can be applied iteratively, starting from $W=\emptyset$, and adding a new generator to $W$ in each iteration. Note that, for termination, the learner would need to have knowledge of the cardinality of $U$, so that it can decide when to stop looking for more generators. Naturally, the success of this learning method crucially hinges on the existence (and effective construction) of ``small'' $(\Sigma^*\setminus W)$-universal sets, whenever $W$ is a strict subset of a potential target generator set $U$. We will see below that such ``small'' $(\Sigma^*\setminus W)$-universal sets exist for example when all potential generator sets $U$ have size 2, but not in general if they have size 3.

One consequence of Lemma~\ref{lem:k-nearlyUniversalMemQ} is that principal shuffle ideals can be learned efficiently using membership queries (see the appendix for a proof):

\begin{theorem}\label{thm:principalMemQ}
    The class of all principal shuffle ideals is learnable from polynomially many membership queries.
\end{theorem}

Another application of Lemma~\ref{lem:k-nearlyUniversalMemQ} is that shuffle ideals generated by exactly two words are efficiently learnable from membership queries. To prove this, we first extend a construction from~\cite{DBLP:conf/dcfs/FleischmannHHMN22} to show that small $(\Sigma^n\setminus\{u\})$-universal sets always exist for $u\in\Sigma^*$ and $n\ge |u|$, i.e., the premise of Lemma~\ref{lem:k-nearlyUniversalMemQ} is fulfilled for $M:=\bigl\{U=\{u,u'\}\mid u\not\sqsubseteq u'\bigr\}$.

\begin{lemma}\label{lem:avoid-one-generator}
    Let $u\in\Sigma^d$ and let $n\ge d$.
    Then there exists a $(\Sigma^n\setminus\{u\})$-universal set
    $S(n,\{u\})=\{s\}\subseteq \Sigma^{\ge n}$ such that, for fixed alphabet $\Sigma$,
    $|s|=O(dn)$.
\end{lemma}

\begin{proof}
    Write $u=u_1\cdots u_d$.
    For each $a\in\Sigma$, let $c_a$ be a word containing each letter of
    $\Sigma\setminus\{a\}$ exactly once, and set $B_a(n)=c_a^n$.
    Define
    $s=B_{u_1}(n)\,u_1\,B_{u_2}(n)\,u_2\cdots u_{d-1}\,B_{u_d}(n)$,
    and let $S(n,\{u\})=\{s\}$.
    
    The word $s$ avoids $u$ as a scattered subword: after the explicit occurrence
    of the prefix $u_1\cdots u_{d-1}$, the remaining suffix is $B_{u_d}(n)$, which
    does not contain the letter $u_d$.
    
    Conversely, let $v\in\Sigma^n$ with $u\not\sqsubseteq v$, and let $j$ be the least
    index such that $u_1\cdots u_j\not\sqsubseteq v$.
    Using the leftmost embedding of $u_1\cdots u_{j-1}$ in $v$, we may write
    $v=x_0u_1x_1u_2\cdots x_{j-2}u_{j-1}y$,
    where each $x_i$ avoids the letter $u_{i+1}$, and the suffix $y$ avoids the
    letter $u_j$.
    Since all these factors have length at most $n$, we have
    $x_i\sqsubseteq B_{u_{i+1}}(n)$ for $0\le i\le j-2$, and
    $y\sqsubseteq B_{u_j}(n)$.
    Therefore $v\sqsubseteq s$.
    This shows that $\{s\}$ is indeed a $(\Sigma^n\setminus\{u\})$-universal set.
    
    Finally, for fixed alphabet $\Sigma$, we have $|s|=O(dn)$. \qed
\end{proof}

\begin{theorem}\label{thm:twoGenerators}
    Let $\mathcal{L}$ be the class of all shuffle ideals of the form $SI(U)$ where $|U|=2$. Then $\mathcal{L}$ is learnable
    from polynomially many membership queries.
\end{theorem}

\begin{proof}
    Let the target language be $SI(\{g,h\})$, where $g$ is one generator first learned
    using Lemma~\ref{lem:k-nearlyUniversalMemQ} with $W=\emptyset$ (using the universal sets available for all $n$).
    For each $n=|g|,|g|+1,|g|+2,\dots$, Lemma~\ref{lem:avoid-one-generator}
    provides a single word $s_n\in\Sigma^{\ge n}\setminus SI(g)$ such that
    $\{s_n\}$ is $(\Sigma^n\setminus\{g\})$-universal and $|s_n|=O(|g|n)$.
    The desired learner queries in order each $s_n$ until the first positive answer is obtained, and then distills
    $s_n$ into a generator using Algorithm~\ref{alg:distilling}.
    
    A positive answer must occur for some $n\le \max(|g|,|h|)$.
    If $|h|\ge |g|$, then $h$ itself is a word avoiding $g$, so
    $s_{|h|}$ contains $h$ as a scattered subword and is therefore positive.
    If $|h|<|g|$, let $x\in\Sigma^{|g|}$ be any extension of $h$ with $x\neq g$.
    Then $h\sqsubseteq x$, so $x\in SI(\{g,h\})$, while $g\not\sqsubseteq x$ because
    $|x|=|g|$ and $x\neq g$.
    Hence $s_{|g|}$ contains $x$, and therefore also $h$, as a scattered subword.
    
    Finally, every queried word avoids $g$, so any positive answer must be due to $h$.
    Therefore the distilling procedure cannot return $g$; it must return the other
    generator $h$.
    
    Since the search over $n$ stops by $n=\max(|g|,|h|)$, and for each such $n$ the
    queried word has polynomial length, the total number of queries is
    polynomial in $|g|+|h|$. \qed
\end{proof}

We can extend this positive result to classes of shuffle ideals with \emph{at most}\/ two generators, as long as they have the same length (Theorem~\ref{thm:<=2generatorsMemQsameLength}). To this end, we first observe the following lemma, which, together with Theorem~\ref{thm:<=2generatorsMemQsameLength}, is proven in the appendix.

\begin{lemma}\label{lem:2ndGeneratorMemQ}
    Suppose the class $\mathcal{L}$ of languages to be learned is any subclass of $\mathcal{SI}$. Suppose further that a learner knows one longest generator $u_1\in\arg\max_{u\in U}|u|$ of the target language $SI(U)$. Then polynomially many membership queries suffice for the learner to (i) determine whether $U$ contains a generator $u_2\ne u_1$ with $|u_2|\le|u_1|$ and (ii) find one such $u_2$ if one exists.
\end{lemma}

\begin{theorem}\label{thm:<=2generatorsMemQsameLength}
    The class of all shuffle ideals of the form $SI(U)$ where $U\subseteq \Sigma^n$ for some $n$ and $1\le |U|\le 2$ is learnable from polynomially many membership queries, without the learner knowing $n$ in advance.
\end{theorem}

Identical length of the two generators is crucial here, as the following result shows. 

\begin{theorem}\label{thm:<=2generatorsMemQdifferentLength}
    Let $\mathcal{L}$ be the class of all shuffle ideals of the form $SI(U)$ where $1\le |U|\le 2$. Then $\mathcal{L}$ is not learnable from membership queries.
\end{theorem}

\begin{proof} \emph{Sketch; see appendix for details.}
    Suppose a learner $A$ learns every language in $\mathcal{L}$ from membership queries. Since no shuffle ideal contains the empty word, $A$ must ask a query for at least one non-empty word. Let $a\in\Sigma$ be any letter occurring in the first non-empty word queried by $A$. One can then show that $A$ fails to distinguish $SI(a)$ from languages $SI\bigl(\{a,b^{n}\}\bigr)$ for sufficiently large $n$. \qed 
\end{proof}

Lemma~\ref{lem:k-nearlyUniversalMemQ} identifies a sufficient condition for efficient exact learning from membership queries. We have seen above that this condition is met in special cases. However, the following negative result shows that it cannot hold for arbitrary forbidden sets $W$.

\begin{theorem}\label{thm:MQexpLower}
    Let $\mathcal{L}$ be the class of all shuffle ideals of the form $SI(U)$ where $|U|=3$. Then $\mathcal{L}$ is not learnable from polynomially many membership queries. 
\end{theorem}

\begin{proof}
    Fix two distinct letters $a,b\in\Sigma$.
    For every $r\ge 2$, let
    $    W_r=\{a^r,b^r\},
    A_r=\{v\in\{a,b\}^{2r-2}\mid |v|_a=|v|_b=r-1\},
    $
    and consider the class
    $
    \mathcal C_r=\bigl\{\,SI(W_r\cup\{v\}) \mid v\in A_r\,\bigr\}.
    $
    
    Every target in $\mathcal C_r$ has exactly three generators:
    $a^r$, $b^r$, and $v$. Indeed, $v$ contains only $r-1$ occurrences of $a$
    and of $b$, so neither $a^r$ nor $b^r$ is a scattered subword of $v$; and
    $v$ contains both $a$ and $b$, so $v \not\sqsubseteq a^r$ and $v\not\sqsubseteq b^r$.
    The size of each target is
    $
    n=|a^r|+|b^r|+|v|=r+r+(2r-2)=4r-2=\Theta(r)
    $.
    Moreover, $|\mathcal C_r|=|A_r|=\binom{2r-2}{r-1},$
    which is exponential in $r$, and hence exponential 
    in $n$.
    
    Let $x$ be a membership query.
    If $x\in SI(W_r)$, the answer is 1 for every concept in
    $\mathcal C_r$.
    Suppose now that $x\notin SI(W_r)$. Then $x$ contains at most $r-1$
    occurrences of $a$ and at most $r-1$ occurrences of $b$.
    If $x\in SI(W_r\cup\{v\})$ for some $v\in A_r$, then 
    $v\sqsubseteq x$. Since $v$ contains exactly $r-1$ occurrences of $a$
    and exactly $r-1$ occurrences of $b$, also $x$ must contain
    exactly $r-1$ occurrences of each of these two letters. Hence $v$ is uniquely
    determined: it is precisely the word obtained from $x$ by deleting all letters
    outside $\{a,b\}$.
    Thus, for a query $x\notin SI(W_r)$, at most one concept in $\mathcal C_r$
    answers positively.
    
    Thus, in the worst case, after $q$ queries, 
    $|\mathcal C_r|-q$ concepts in $\mathcal C_r$ remain consistent with the
    answers: indeed, as long as at least two concepts remain,
    one may choose a target that answers $0$ to the next query, thereby
    eliminating at most one candidate. Hence, 
    at least $|\mathcal C_r|-1$ queries are needed in the worst case. This number  is exponential in $n$. \qed 
\end{proof}

Combining Lemma~\ref{lem:k-nearlyUniversalMemQ} with the proof of Theorem~\ref{thm:MQexpLower}, shows that small $(\Sigma^n\setminus W)$-universal sets, for $W\subseteq U\in M$, cannot be constructed in the sense of the premise of Lemma~\ref{lem:k-nearlyUniversalMemQ}, if $k=3$ and $M=\{\{a^r,b^r,v\}\mid r\ge 2, v\in\{a,b\}^{2r-2},|v|_a=|v|_b=r-1\}$.

By contrast with Theorem~\ref{thm:MQexpLower}, Angluin~\cite{DBLP:journals/iandc/Angluin87} showed that the class of all regular languages is learnable from polynomially many membership and \emph{equivalence queries}. Given a class $\mathcal{L}$ and a target language $L\in\mathcal{L}$, an equivalence query is a language $L'\in\mathcal{L}$. If $L'=L$, the query is answered 1, and the learning process ends. If $L'\ne L$, the answer is of the form $(0,w)$, where $w\in\Sigma^*$ is a \emph{counterexample}\/ in the symmetric difference of $L'$ and $L$. Although $\mathcal{SI}$ is a subclass of the class of all regular languages, one cannot immediately transfer Angluin's result to conclude that $\mathcal{SI}$ is learnable from polynomially many membership and equivalence queries. This is because the learning algorithm can only ask equivalence queries that refer to languages $L'$ contained in the underlying class $\mathcal{L}$. We therefore provide a proof for the following result (see appendix):

\begin{theorem}
The class $\mathcal{SI}$ is learnable from polynomially many membership and equivalence queries. In particular, there is a learner $A$ learning any language $SI(U)\in\mathcal{SI}$ with $O(m)$ equivalence queries and $O(n^2+mz^2)$ membership queries, where $m=|U|$ and $n=\min\{|u|\mid u\in U\}$ are not known to $A$ in advance, and $z$ is the length of the longest counterexample provided to $A$ in the learning process.
\end{theorem}

\section{Learning Shuffle Ideals from Contrastive Queries}

Principal shuffle ideals are simple enough to already be learnable from polynomially many membership queries. If the learner can ask contrastive queries, even a single query suffices, for a wide variety of distance functions (see the appendix for a proof):

\begin{proposition}\label{prop:principalConQ}
Let $d$ be either length-compatible or the edit distance. 
Then the  class of principal shuffle ideals over $\Sigma$ is learnable from a single contrastive query w.r.t.\ $d$.
\end{proposition}

However, without any assumptions on the form of the generators, the class of all shuffle ideals with at most two generators is not learnable with contrastive queries, for the same broad class of natural distance functions as considered in Proposition~\ref{prop:principalConQ}:

\begin{theorem}\label{thm:<=2generatorsConQdifferentLength}
    Let $\mathcal{L}$ be the class of all shuffle ideals of the form $SI(U)$ where $1\le |U|\le 2$. Let $d$ be either the edit distance or any length-compatible distance function. Then $\mathcal{L}$ is not learnable from contrastive queries w.r.t.\ $d$.
\end{theorem}

\begin{proof}
    Suppose a learner $A$ learns every language in $\mathcal{L}$ from contrastive queries w.r.t.\ $d$. Let $a,b\in\Sigma$, $a\ne b$.
Consider $A$ interacting with an oracle defined as follows. 

\smallskip
   (1) If $A$ queries $\varepsilon$, then the oracle answers $(0,a)$. 
   
   (2) If $A$ queries a word $w$ containing the symbol $a$, then the oracle responds $(1,w')$ where $|w'|=|w|=:n$, $w'$ does not contain the symbol $a$, and $w'$ minimizes $d(w,w')$ among all words without the symbol $a$. (If $d$ is length-compatible, all minimizers of $d(w,w')$ not containing $a$ have length $n$; such minimizers exist. If $d$ is the edit distance, then changing all occurrences of $a$ in $w$ to $b$ yields the desired form of $w'$.)
    
    (3) If $A$ queries a word $w$ without the symbol $a$, then the oracle responds $(0,w')$ where $|w'|=|w|=n$, $w'$ contains the symbol $a$, and $w'$ minimizes $d(w,w')$ among all words with the symbol $a$. (If $d$ is length-compatible, all minimizers of $d(w,w')$ containing $a$ have length $n$; such minimizers exist. If $d$ is the edit distance, then changing a single  occurrence of a symbol from $\Sigma\setminus\{a\}$ in $w$ to $a$ yields the desired form of $w'$.)

    \smallskip
Since this oracle answers $A$'s queries consistently with the potential target $SI(\{a\})$, the learner $A$ must eventually halt and conjecture $SI(\{a\})$. Now let $m$ be the length of the longest thus word queried by $A$ before halting. It follows that $A$'s interaction with the oracle for $SI(\{a\})$ is also fully consistent with the potential target language $SI(\{a,b^{m+1}\})$. Hence $A$ fails to identify $SI(\{a,b^{m+1}\})$, which belongs to $\mathcal{L}$. \qed
\end{proof}

By contrast, Theorem~\ref{thm:poly-fixed-length} will  show that, when all generators of any single shuffle ideal are of the same length, efficient learning from contrastive queries is possible. 
It remains open whether polynomially many membership queries alone are sufficient. As in the proof of Lemma~\ref{lem:k-nearlyUniversalMemQ}, polynomially many membership queries will be sufficient in case a general method for constructing ``small'' $(\Sigma^n\setminus W)$-universal sets exists. Theorem~\ref{thm:poly-fixed-length} establishes a positive learnability result \emph{without}\/ relying on $(\Sigma^n\setminus W)$-universal sets. This requires the definition of a specific (yet natural) distance function.

\begin{definition}
    Fix an arbitrary total order on $\Sigma$, and for every $n$ consider the lexicographic order on $\Sigma^n$ induced by this order.
    For any word $w\in\Sigma^n$, let $\mathrm{idx}(w)$ be the zero-based position of $w$ in this ordering of $\Sigma^n$. Now the shortlex distance function $d_{\mathrm{slex}}$ over $\Sigma^*$ is defined as follows, for any $x,y\in\Sigma^*$: 
    \[
d_{\mathrm{slex}}(x,y) = \begin{cases}\bigl| |x| - |y| \bigr|\,,&\mbox{if }|x|\ne |y|\,,\\
\frac{\bigl| \mathrm{idx}(x) - \mathrm{idx}(y) \bigr|}{|\Sigma^{|x|}|}\,,&\mbox{if }|x|=|y|\,.
\end{cases}
\]
\end{definition}

An oracle providing contrastive examples w.r.t.\ $d_{\mathrm{slex}}$ strictly prioritizes minimizing the length difference, using the (lexicographic) index distance solely as a tie-breaker for words of equal length.

Given $B\subseteq \Sigma^n$, a \emph{lexicographic gap of $B$ in $\Sigma^n$} is a maximal lexicographic interval in $\Sigma^n\setminus B$.
In particular, if $|B|=t$, then $\Sigma^n\setminus B$ is a disjoint union of at most $t+1$ gaps.
For instance, if $B=\varnothing$, the only gap is $\Sigma^n$. If $B=\{abb\}$, the gaps are the words in $\Sigma^n$ that are strictly smaller than $abb$ and the words in $\Sigma^n$ that are strictly greater than $abb$.

\begin{theorem}\label{thm:poly-fixed-length}
    Let $\mathcal{L}$ be the class of shuffle ideals for which no two generators differ in length, i.e., 
    $\mathcal{L}=\{SI(U)\mid\exists n\ge 1\ U\subseteq \Sigma^n\}$.
    Then $\mathcal{L}$ is learnable using $O(|U|^2)$ contrastive queries w.r.t.\ $d_{\mathrm{slex}}$. This is witnessed by a learner with polynomial runtime.
\end{theorem}

\begin{proof}
    Assume a subset $B\subseteq U$ of generators is already known.
    As $U\subseteq \Sigma^n$, the positive words of length $n$ are exactly the generators in $U$. Hence the unknown generators $U\setminus B$ are in the complement $\Sigma^n\setminus B$, which is a disjoint union of $\le |B|+1$ lexicographic gaps.
    
    We show that there is a polynomial-time procedure using at most $2|B|+1$ contrastive queries that either outputs a word in $U\setminus B$ or correctly concludes $B=U$.
    We process each gap by its interval of lexicographic indices $[\ell,r]$.
Let
$s=\left\lceil(\ell+r)/2\right\rceil$,
and let $q$ be the word in $\Sigma^n$ with $\mathrm{idx}(q)=s$.
Let $L$ be the target language. Consider three cases:

    \textbf{Case 1: $q\in L$.}
    Since $|q|=n$, we have $q\in U$.
    As $\mathrm{idx}(q)=s\in[\ell,r]$ and the current gap contains no known generators, it follows that $q\in U\setminus B$.

    \textbf{Case 2: $q\notin L$ and the oracle returns a word $v\in L\setminus B$.}
    By definition of $d_{\mathrm{slex}}$, the oracle  returns a word of length $|q|$, i.e., $|v|=n$.   
    Hence $v\in U$, and thus $v\in U\setminus B$.

    \textbf{Case 3: $q\notin L$ and the oracle returns a known generator $v\in B$.}
    Let $t=\mathrm{idx}(v)$. Let $p=2s-t$ be the point at the same distance from $q$ as $v$, on the opposite side of $q$.
    
    If $t<\ell$, then $p>r$.
    Hence every index in $[\ell,r]$ is strictly closer to $s$ than $t$ is, so no unknown generator can lie in the gap.
    
    If $t>r$, then $p\le \ell$.
    Any unknown generator in $U\setminus B$ with index $x\in [\ell,r]$ satisfies
    $|x-s|\ge |t-s|=t-s$, else it would be strictly closer to $q$ than $v$.
    Since $x\le r<t$, we have $x\le s$, otherwise $|x-s|=x-s\le r-s<t-s$, 
    contradicting $|x-s|\ge t-s$. 
    Therefore $s-x=|x-s|\ge t-s$, so $x\le 2s-t=p$.
    Thus, if $p<\ell$, the gap is empty; if $p=\ell$, then only the single point $\ell$ can still contain an unknown generator.
    In the latter case, let $y$ be the unique word with $\mathrm{idx}(y)=\ell$.
    If $q=y$, then $y$ has already been queried negatively, so the gap is empty.
    Otherwise, we query $y$: if the query is positive, we have found a new generator; otherwise the gap is empty.

\smallskip
    
    In Cases (1) and (2) we output a new generator. In Case (3) we discard the gap (with at most an additional query). If all gaps are discarded, then every point of $\Sigma^n\setminus B$ has been shown not to belong to $U\setminus B$, so $U\setminus B=\varnothing$, and hence $U=B$.
    
   We recover the first generator $u\in U$ (and thus the value of $n$) by querying $\varepsilon$. 
   We now iteratively recover $U$, starting with $B=\{u\}$. At each step, we apply the procedure above to the current set $B$. If the procedure returns a word $w\in U\setminus B$, we update $B\gets B\cup\{w\}$ and repeat. If it concludes that $B=U$, we terminate and output $L=SI(B)$.
    
    Each successful iteration adds a new generator to $B$, so there are at most $|U|$ such iterations. In addition, one final iteration is needed for concluding $B=U$. Each iteration uses polynomial time and at most $2|B|+1\le 2|U|+1$ queries. In total, the running time is polynomial in $n$ and $|U|$, and the number of queries is $O(|U|^2)$. \qed
\end{proof}

The procedure in the proof of Theorem~\ref{thm:poly-fixed-length}, which we call the \emph{midpoints-over-gaps}\/ strategy, can also be applied to identify generators at various lengths. If a set $W$ of generators of length at most $n$ has been found, and one wants to find a generator at length $n'\ge n$, then all words in $\Sigma^{n'}$ that contain any word in $W$ as a scattered subword form the initial set $B$. The learning process then proceeds exactly as described in the proof of Theorem~\ref{thm:poly-fixed-length}. This leads to the following result, where the claimed query efficiency is due to the fact that the number of gaps to be searched 
is bounded by a polynomial; see the appendix for a proof.

\begin{theorem}\label{thm:fewgaps2}
    Let $k\in\mathbb{N}$ be a constant and let $\mathcal{L}_k$ be 
    the class of shuffle ideals $SI(U)$ for which $\max_{u\in U}|u| - \min_{u\in U}|u|\le k$. Then $\mathcal{L}_k$ is learnable with polynomially many contrastive queries w.r.t.\ $d_{\mathrm{slex}}$.
\end{theorem}
This raises the question how complex a class of shuffle ideals can be learned efficiently with the midpoints-over-gaps strategy. Theorem~\ref{thm:<=2generatorsConQdifferentLength} already tells us that not every class of shuffle ideals can be learned at all this way or efficiently this way. The intrinsic reasons are twofold: (1) if there is no bound on the length of generators, the learner may not know at which length to stop searching for more generators with the midpoints-over-gaps strategy; (2) the number of gaps grows exponentially with the length of the generators (even for $|\Sigma|=2$), as Proposition~\ref{prop:binaryIntervalsComplement} (proven in the appendix) shows.

\begin{proposition}\label{prop:binaryIntervalsComplement}
    Let $\Sigma$ be a binary alphabet, let
    $u\in\Sigma^k$ contain 
    two distinct symbols, and let $m\ge k$. 
    Then, $\Sigma^m\cap SI(u)$ has exactly $\binom{m-1}{k-1}+1$ lexicographic gaps.
\end{proposition}

For $k\approx\frac{m}{2}$, the value of $\binom{m-1}{k-1}+1$ grows exponentially in $m$, so that searching for generators at arbitrary lengths $m$ will cost the midpoints-over-gaps strategy an exponential number of queries. However, as Theorems~\ref{thm:poly-fixed-length} and~\ref{thm:fewgaps2} demonstrate, some subclasses of shuffle ideals impose conditions stringent enough to result in a polynomial number of lexicographic intervals to be checked by the midpoints-over-gaps strategy. 



Finally, we add a simple observation. 
Note that the complement of any shuffle ideal contains the empty string. Even if we consider the complement w.r.t.\ $\Sigma^+$, it is easy to distinguish shuffle ideals from complements of shuffle ideals by contrastive queries w.r.t.\ any distance function (see the appendix). 
We thus conclude, using Corollary~\ref{cor:complement}:

\begin{corollary}\label{cor:complementApplied}
    Let $\mathcal{L}$ be any of the classes 
    from Proposition~\ref{prop:principalConQ}, Theorem~\ref{thm:poly-fixed-length}, and Theorem~\ref{thm:fewgaps2}. Then $\mathcal{L}\cup\overline{\mathcal{L}}$ is learnable from polynomially many contrastive queries w.r.t.\ $d_{\mathrm{slex}}$.
\end{corollary}

\section{Conclusions}

Shuffle ideals, despite their appealingly simple description, are difficult to learn from membership queries. We showed (not surprisingly) that they are efficiently learnable from a combination of membership and equivalence queries, but when using membership queries alone, shuffle ideals with at most two generators are not learnable, and those with exactly three generators are not efficiently learnable. However, we derived a sufficient condition for efficient learnability, based on our proposed notion of $(\Sigma^*\setminus W)$-universal sets---a notion possibly of independent interest. An open question is for which collections 
of generator sets 
one can effectively construct small $(\Sigma^*\setminus W)$-universal sets. 

Our midpoints-over-gaps strategy showcases the potential of contrastive queries; a second open problem is to find further classes of shuffle ideals for which this strategy is provably efficient.

Corollary~\ref{cor:complementApplied} extends our learnability results on contrastive queries to include complements of shuffle ideals. While this result is trivial and only a small step towards learning of piecewise testable languages, one direction of future work is to extend our learning strategies to special cases of Boolean combinations of shuffle ideals.

\begin{credits}
\subsubsection{\discintname}
The authors have no competing interests to declare.
\end{credits}
%
%
%
 \bibliographystyle{splncs04}
 \bibliography{biblio}

 \appendix

 \section{Proofs Omitted from the Main Body}

 \setcounter{proposition}{0}
\begin{proposition}
   Let $\mathcal{L}$ be any class of languages over $\Sigma^*$ and $d$ a distance function. Then $\mathcal{L}$ is learnable from (polynomially many) contrastive queries w.r.t.\ $d$ iff $\overline{\mathcal{L}}$ is learnable from (polynomially many) contrastive queries w.r.t.\ $d$ (where the size of $\overline{SI(U)}$ is $\|U\|$). 
\end{proposition}

\begin{proof}
    It suffices to show one direction. Let $A$ be a successful learner for $\mathcal{L}$. A learner for $\overline{\mathcal{L}}$ only needs to simulate $A$, pass every query by $A$ to the oracle, and invert all the labels returned by the oracle before passing them on to $A$. \qed
\end{proof}

\begin{proposition} Let $\mathcal{L}$ be a class of languages and $d_1,d_2:\Sigma^*\times\Sigma^*\rightarrow\mathbb{R}$ two functions such that $d_1$ is a refinement of $d_2$. If $\mathcal{L}$ can be learned from (polynomially many) contrastive queries  using the distance function $d_2$, then it can be learned from (polynomially many) contrastive queries using the distance function $d_1$.
\end{proposition}
\begin{proof}
An oracle that is truthful w.r.t.\ $d_1$, given a query $w$, must return an example $w''\in\arg\min_{L(w')\ne L(w)}d_1(w,w')$. Clearly, $\arg\min_{L(w')\ne L(w)}d_1(w,w')\subseteq \arg\min_{L(w')\ne L(w)}$ $d_2(w,w')$, so that $w''\in\arg\min_{L(w')\ne L(w)}d_2(w,w')$. Therefore, the same learner that learns $\mathcal{L}$ from (polynomially many) contrastive queries using the distance function $d_2$, also learns $\mathcal{L}$ from (polynomially many) contrastive queries using~$d_1$. \qed
\end{proof}

\setcounter{lemma}{1}
\begin{lemma}
Let $M$ be a set of finite subsets of $\Sigma^*$, where no $U\in M$ contains words $u\!\ne\! u'$ with $u\!\sqsubseteq \! u'$. 
Suppose there is an algorithm with the following input/output behavior.\\
Input: $(n,W)$ where 
$n\le\max_{u\in U} |u|$ and  $W\subset U$ for some unknown $U\in M$.\\
Output: a $(\Sigma^n\setminus W)$-universal set $S(n,W)\subseteq \Sigma^{\ge n}$ such that $\sum_{s\in S(n,W)}|s|$ is bounded from above by a polynomial in $(n,|W|)$. 

Then there is a learner that, for any target language $SI(U)$ with $U\in M$, and on input of any $W\subset U$, finds a generator $u\in U\setminus W$, using polynomially many membership queries in $\|U\|$.
\end{lemma}

\begin{proof}
In order to find the new generator $u$, the learner first asks all words in a $(\Sigma^n\setminus W)$-universal set, for $n=1,2,\ldots$ until it receives the answer 1 for the first time, say for query $w$. Then
the learner distills $w$ into a generator $u$, using Algorithm~\ref{alg:distilling}.

The number of queries asked by the learner is polynomial, since, by the premise of the lemma, the size of the $(\Sigma^n\setminus W)$-universal sets, as well as the lengths of the words therein, are polynomial.
The number of membership queries asked before the first positive answer $u$ is at most
$\sum_{n=1}^{|u|} |S(n,W)| \le \sum_{n=1}^{|u|} \sum_{s\in S(n,W)} |s|$, which is bounded by a polynomial in $\|U\|$. \qed
\end{proof}

\setcounter{theorem}{0}
\begin{theorem}
    The class of all principal shuffle ideals is learnable from polynomially many membership queries.
\end{theorem}

\begin{proof}
    By Lemma~\ref{lem:k-nearlyUniversalMemQ} (with $W=\emptyset$), it suffices to provide an algorithm that, given $n\in\mathbb{N}$, constructs a set $S(n)\subseteq \Sigma^{\ge n}$ of words such that $\sum_{s\in S(n)}|s|$ is bounded from above by a polynomial in $n$,  and $\bigcup_{s\in S(n)}(\Sigma^n\cap\{v\mid v\sqsubseteq s\})=\Sigma^n$. But this is easy to do; the set $S$ consists simply of the single string $s_n:=(\sigma_1\cdots\sigma_{|\Sigma|})^n$, where $\Sigma=\{\sigma_1,\ldots,\sigma_{|\Sigma|}\}$. Clearly, $|s_n|=|\Sigma|\cdot n\in O(n)$ and $s_n$ contains all words in $\Sigma^n$ as scattered subwords. \qed
\end{proof}

\setcounter{lemma}{3}
\begin{lemma}
    Suppose the class $\mathcal{L}$ of languages to be learned is any subclass of $\mathcal{SI}$. Suppose further that a learner knows one longest generator $u_1\in\arg\max_{u\in U}|u|$ of the target language $SI(U)$. Then polynomially many membership queries suffice for the learner to (i) determine whether $U$ contains a generator $u_2\ne u_1$ with $|u_2|\le|u_1|$ and (ii) find one such $u_2$ if one exists.
\end{lemma}

\begin{proof}
Let $n=|u_1|$ be the length of the longest generator of the target language $L$. By Lemma~\ref{lem:avoid-one-generator} with $n=d$, there exists a $(\Sigma^n\setminus\{u_1\})$-universal set
    $\{w\}\subseteq \Sigma^{\ge n}$, where the length of $w$ is polynomial in $n$. (This can alternatively be concluded from a result in \cite{DBLP:conf/dcfs/FleischmannHHMN22}, by which one can construct a word $w$ of length \emph{linear} in $n$ such that $w$ contains all words of $\Sigma^n$ except $u_1$ as a scattered subword.) 

The desired learner now only needs to query $w$. If the answer is 0, clearly $L$ has no generator of length at most $n$, besides $u_1$. If the answer is 1, $L$ must have a second generator of length at most $n$. Such a generator can now be found by distilling $w$ into a generator, using Algorithm~\ref{alg:distilling}. \qed 
\end{proof}

\setcounter{theorem}{2}
\begin{theorem}
    The class of all shuffle ideals of the form $SI(U)$ where $U\subseteq \Sigma^n$ for some $n$ and $1\le |U|\le 2$ is learnable from polynomially many membership queries, without the learner knowing $n$ in advance.
\end{theorem}

\begin{proof}
    The same learning method as deployed in the proof of Theorem~\ref{thm:principalMemQ} can be used to identify one generator of the target language. If a second generator exists, it has the same length as the first one; hence we can invoke Lemma~\ref{lem:2ndGeneratorMemQ} to complete the proof. \qed 
\end{proof}

\begin{theorem}
    Let $\mathcal{L}$ be the class of all shuffle ideals of the form $SI(U)$ where $1\le |U|\le 2$. Then $\mathcal{L}$ is not learnable from membership queries.
\end{theorem}

\begin{proof}
    Suppose there is a learning algorithm $A$ that learns every language in $\mathcal{L}$ from membership queries. Since no shuffle ideal contains the empty word, $A$ must ask a query for at least one non-empty word. Let $a\in\Sigma$ be any letter occurring in the first non-empty word queried by $A$. 

    Consider the case where the target language is $SI(a)$; by definition, $A$ asks at most a finite number of further queries before halting with the conjecture $SI(a)$. The oracle answers all queries containing the letter $a$ with 1 and all others with 0.
    Let $n$ be the length of the longest word queried in this process, and let $b\in\Sigma$ be any letter distinct from $a$. Clearly, all responses by the oracle are consistent with the potential target language $SI\bigl(\{a,b^{n+1}\}\bigr)$. Thus $A$ fails to identify $SI\bigl(\{a,b^{n+1}\}\bigr)$---a contradiction. \qed 
\end{proof}

\setcounter{theorem}{5}
\begin{theorem}
The class $\mathcal{SI}$ is learnable from polynomially many membership and equivalence queries. In particular, there is a learner $A$ learning any language $SI(U)\in\mathcal{SI}$ with $O(m)$ equivalence queries and $O(n^2+mz^2)$ membership queries, where $m=|U|$ and $n=\min\{|u|\mid u\in U\}$ are not known to $A$ in advance, and $z$ is the length of the longest counterexample provided to $A$ in the learning process.
\end{theorem}

\begin{proof}
  Let $\Sigma=\{\sigma_1,\ldots,\sigma_{|\Sigma|}\}$ and suppose $SI(U)$ is the unknown target language. 
  In order to find a first generator in $U$, the learner $A$ asks  membership queries of the form $s_{n'}:=(\sigma_1\cdots\sigma_{|\Sigma|})^{n'}$,  for $n'=1,2,3,\ldots$, until it receives the answer 1 for the first time, say for query $s_{n^*}$. Then $s_{n^*}$ contains a generator from $U$ as a scattered subword. 
  
  Now, 
  $A$ distills $s_{n^*}$ into a generator $u\in U$, where Algorithm~\ref{alg:distilling} is used as a subroutine and $u$ is its output. (Note that $u$ may not be a shortest generator in $U$, but $n^*$ is at most the length $n$ of a shortest generator in $U$.)
  So far, $A$ has consumed $O(|\Sigma|^2n^2)$ membership queries in total.
  Next, $A$ sets $U^*=\{u\}$ and executes the following instructions. (Throughout this process, $SI(U^*)$ will be a subset of the target language $L$.)
  \begin{enumerate}
      \item Ask an equivalence query for $SI(U^*)$. 
      \item If the answer is 1, stop and output $L=SI(U^*)$. If the answer is $(0,w)$, then $w\in L\setminus SI(U^*)$, $w$ contains a generator $u'\in U\setminus SI(U^*)$ as a scattered subword, and $w$ does not contain any generator from $U^*$ as a scattered subword, since $SI(U^*)\subseteq L$. 
      \item Distill $w$ into a generator $u'$. Since $w$ does not contain any element of $U^*$ as a scattered subword, we obtain $u'\in U\setminus U^*$, i.e., $u'$ is a newly found generator.
      \item Add $u'$ to $U^*$ and goto 1.
  \end{enumerate}
  For each of the $m-1$ equivalence queries, distillation costs $O(z^2)$ membership queries, for a total of $O(mz^2)$ membership queries after the first generator is found. Overall, the learner has consumed $O(n^2+mz^2)$ membership queries. This proves the theorem. \qed 
\end{proof}

\begin{proposition}
Let $d$ be either the edit distance or any length-compatible distance function. Then the  class of principal shuffle ideals over $\Sigma$ is learnable from a single contrastive query w.r.t.\ $d$.
\end{proposition}

\begin{proof} Both for edit distance and for $d_0$, a learning algorithm can first ask a query for the empty word; the contrastive example will be a shortest string in the target language $L$, which is the unique generator of $L$.\footnote{One might argue that querying $\varepsilon$ gives the learner an unfair advantage; the learner knows in advance that $\varepsilon\notin L$, so perhaps queries ought to be restricted to $\Sigma^+$. In this case, under $d_0$, the learner only needs to ask a single query for a word $a$ for some $a\in\Sigma$ and will learn the shortest generator from the oracle's response. Under the edit distance, the learner would have to ask every word $a\in\Sigma^1$, increasing the number of queries to a constant $|\Sigma|$.} Thus the claim follows with Corollary~\ref{cor:refinement}. \qed
\end{proof}

\begin{proposition}
     Let $\Sigma$ be a binary alphabet, let
 $u\in\Sigma^k$ contain 
 two distinct symbols, and let $m\ge k$. 
 Then, $\Sigma^m\cap SI(u)$ has exactly $\binom{m-1}{k-1}+1$ lexicographic gaps.
 \end{proposition}

\begin{proof} Let $\Sigma=\{a,b\}$ with $a<b$, and let $u$ contain both the symbol $a$ and the symbol~$b$.

    Let $T(u,m)=SI(u)\cap\Sigma^m$. A lexicographic interval in $\Sigma^m$ is a set of consecutive words with respect to the fixed lexicographic order. The set $T(u,m)$ can therefore be partitioned uniquely into maximal lexicographic intervals; let $I(u,m)$ denote their number.
    We first show that $I(u,m)=\binom{m-1}{k-1}$ for every $u\in\Sigma^k$.
    
    Write $u=cv$, where $c\in\{a,b\}$, $v\in\Sigma^{k-1}$, and let $\bar c$ denote the other letter in $\Sigma$. A word in $\Sigma^m$ belongs to $T(cv,m)$ if and only if either it starts with $c$ and its suffix of length $m-1$ contains $v$ as a scattered subword, or it starts with $\bar c$ and its suffix contains $cv$ as a scattered subword. Hence $T(cv,m)=c\,T(v,m-1)\cup \bar c\,T(cv,m-1)$. 
    Prefixing by a fixed letter preserves the lexicographic order within the corresponding first-letter block, so the two components contribute $I(v,m-1)$ and $I(cv,m-1)$ maximal intervals, respectively. The only possible merger between them could occur across the boundary between $a\Sigma^{m-1}$ and $b\Sigma^{m-1}$. The two words at this boundary are $ab^{m-1}\qquad\text{and}\qquad ba^{m-1}.$ 
    They cannot both belong to $T(cv,m)$, since no binary word $cv$ of length at least $2$ is a scattered subword of both. Hence no maximal interval can cross this boundary, and therefore $I(cv,m)=I(v,m-1)+I(cv,m-1)$.    
    
    By induction, $I(u,m)$ depends only on $k=|u|$. With a slight abuse of notation, we may therefore write it as $I(k,m)$. Then $I(k,m)=I(k-1,m-1)+I(k,m-1)$ for $m>k\ge 2$, with boundary conditions $I(1,m)=1$ for all $m\ge 1$ and {$I(k,k)=1$} for all $k\ge 1$. Thus $I(k,m)=\binom{m-1}{k-1}$.
    
    Now, since $u$ contains both $a$ and $b$, neither $a^m$ nor $b^m$ belongs to $SI(u)\cap\Sigma^m$. Hence $T(u,m)$ contains neither the lexicographically smallest nor the largest word of $\Sigma^m$. Therefore its complement in $\Sigma^m$ consists of exactly one more maximal lexicographic interval than $T(u,m)$, i.e., $\binom{m-1}{k-1}+1$. By definition, these maximal intervals of the complement are precisely the lexicographic gaps of $SI(u)\cap\Sigma^m$.\qed
\end{proof}

\setcounter{theorem}{8}
\begin{theorem}
    Let $k\in\mathbb{N}$ be a constant and let $\mathcal{L}_k$ be 
    the class of shuffle ideals $SI(U)$ for which $\max_{u\in U}|u| - \min_{u\in U}|u|\le k$. Then $\mathcal{L}_k$ is learnable with polynomially many contrastive queries w.r.t.\ $d_{\mathrm{slex}}$. 
    \end{theorem}
    \begin{proof}
       The shortest generator of length, say, $n$ can be obtained by querying $\varepsilon$. 
The learner then searches for other generators in increasing order of their length. 
Only lengths up to $n+k$ need to be considered.

Suppose that the learner is currently searching for generators of length $m\ge n$ and has already found all shorter generators. These shorter generators determine a set $B_m\subseteq\Sigma^m$ of words that are already known to be positive; generators of length $m$ found so far are included in $B_m$ as well. Hence any positive word outside $B_m$ is a new generator of length $m$.

We claim $B_m$ has only polynomially many lexicographic gaps. Indeed, every previously found generator has length at most $k$ less than $m$, so, since  $k$ is fixed, it can occur as a scattered subword of only polynomially many words of length $m$. Thus $|B_m|$, and consequently the number of its gaps, is polynomially bounded.

The learner can therefore apply the midpoints-over-gaps procedure of Theorem~\ref{thm:poly-fixed-length} to these gaps. Each application either finds a new generator of length $m$ or establishes that no further such generator exists. Iterating over the relevant lengths identifies all generators using polynomially many contrastive queries and polynomial time. \qed
    \end{proof}
    
Before \textbf{Corollary \ref{cor:complementApplied}}, we claimed that it is easy to \textbf{distinguish shuffle ideals from complements of shuffle ideals} by contrastive queries w.r.t.\ any distance function, even if we consider the complement w.r.t.\ $\Sigma^+$. To see this, consider a learner that asks queries for all words of the form $a_i$, where $a_i\in\Sigma$, as well as, possibly, one query of the form $a_ia_j$ where $a_i\in L$ and $a_j\notin L$. If no $a_i$ belongs to $L$, then $L$ is a shuffle ideal iff a positive contrastive example (of length $>1$) is given. If all $a_i$ belong to $L$, then $L$ is a shuffle ideal iff no negative contrastive example is given. If $a_i\in L$ and $a_j\notin L$, then $L$ is a shuffle ideal iff $a_ia_j\in L$. 

\end{document}